\documentclass[journal]{IEEEtran}

\usepackage[export]{adjustbox}
\usepackage[dvipsnames]{xcolor}
\usepackage{empheq}
\usepackage{amsmath,amsfonts,amsthm,amssymb,amsbsy,paralist,ifthen}
\usepackage{mathrsfs}
\usepackage[mathcal]{euscript}
\usepackage{graphicx}
\usepackage{psfrag}
\usepackage{subfigure}
\usepackage{url}
\usepackage{stfloats}
\usepackage{float}
\usepackage{array}
\usepackage{booktabs} 
\usepackage{algorithm} 
\usepackage{algorithmic} 
\usepackage{color}
\usepackage{cases}
\usepackage{bm}
\usepackage{colortbl}
\usepackage{arydshln}
\usepackage{multirow}
\usepackage{multicol}

\newtheorem{Lem}{Lemma}

\begin{document}

\title{LAURA: LLM-Assisted UAV Routing for AoI Minimization }




\author{Bisheng Wei, Ruichen Zhang, Ruihong Jiang, Mugen Peng, \IEEEmembership{Fellow, IEEE}, Dusit Niyato, \IEEEmembership{Fellow, IEEE}

\thanks{This work was supported in part by NSFC under Grant no. 62301077.}

\thanks{B. Wei, R. Jiang and M. Peng are with the State Key Laboratory of Networking and Switching Technology, Beijing University of Posts and Telecommunications, Beijing 100876, China, (e-mail: weibs@bupt.edu.cn, rhjiang@bupt.edu.cn, pmg@bupt.edu.cn)}

\thanks{R. Zhang, and D. Niyato are with the College of Computing and Data Science, Nanyang Technological University, Singapore (e-mail: ruichen.zhang@ntu.edu.sg, dniyato@ntu.edu.sg).}

}

\maketitle

\begin{abstract}

With the rapid growth of the low-altitude economy, demand for real-time data collection via UAV-assisted wireless sensor networks (WSNs) is increasing. This paper studies the problem of minimizing the age of information (AoI) in UAV-assisted WSNs by optimizing the UAV flight routing. We formulate the AoI minimization task and propose a large language model (LLM)-assisted UAV routing algorithm (LAURA). LAURA employs an LLM as the intelligent crossover operators within an evolutionary optimization framework to efficiently explore the solution space. Simulation results show that LAURA outperforms benchmark methods in reducing the maximum AoI, especially in scenarios with many sensor nodes.

\end{abstract}


\section{Introduction}

The rapid development of the low-altitude economy is transforming various sectors, including urban logistics, infrastructure monitoring, and environmental surveillance \cite{background_1}. Unmanned aerial vehicles (UAVs) have emerged as a critical technology, playing a pivotal role in supporting wireless sensor networks (WSNs) for real-time data collection from spatially distributed ground sensor nodes (SNs). 
A crucial challenge in such systems is maintaining data freshness, especially for time-sensitive applications such as industrial IoT and precision agriculture. To address this, the age of information (AoI) has been proposed as a key metric that quantifies the time elapsed between data generation at a sensor and its reception at the data center \cite{Aoi}. Unlike traditional metrics such as latency or throughput, AoI reflects how current the received data is, directly impacting its decision-making value. Since the UAV's flight routing determines the visiting sequence of SNs, it plays a critical role in AoI performance: different routings result in different transmission orders and time delays, ultimately affecting the freshness of the collected data.

Various approaches have been proposed to address AoI-aware UAV routing planning. For instance, in \cite{AoI_dp}, the UAV routing was optimized using dynamic programming and ant colony heuristic algorithms to minimize the average AoI of data collected from all ground SNs. In \cite{Aoi_GA}, the authors proposed a method that first performed SN association based on affinity propagation clustering, and then used dynamic programming or a genetic algorithm to optimize the UAV routing. 
In \cite{AoI_Transformer}, the authors exploited the transformer to design a machine learning algorithm that optimized the UAV visiting sequence of ground clusters, effectively minimizing the total AoI in IoT networks. While effective, these approaches often involve high computational complexity and limited generalization, particularly when dealing with dynamic or large-scale SN deployments.

Fortunately, large language models (LLMs) have demonstrated impressive capabilities in understanding, generating, and reasoning with natural language \cite{LLM_Ability}. Their embedded knowledge and reasoning ability have attracted increasing interest in their integration into optimization frameworks for solving complex tasks. For instance, \cite{LLM_cloud} introduced a novel edge inference framework for LLMs, incorporating batching and model quantization techniques to ensure high-throughput inference on resource-limited edge devices. \cite{LLM_sat} employed an LLM agent to formulate and solve transmission strategies via mixture-of-experts approach in satellite communication systems. However, these works have not addressed UAV routing optimization, particularly in the context of AoI-driven WSNs.

To fill this gap, we propose an LLM-assisted UAV routing algorithm (LAURA) to minimize AoI. Unlike existing methods such as LEDMA \cite{ISAC_LLM}, which rely on normalization strategies suited for continuous variables, LAURA is capable of handling both continuous and discrete optimization, 
making it well-suited for UAV routing tasks. Our main contributions are as follows: Firstly, we formulate the AoI minimization problem and propose LAURA to solve it for UAV-assisted WSNs. Secondly, we leverage an LLM as intelligent crossover operators within an evolutionary optimization framework, enabling efficient solution exploration without handcrafted heuristics. Thirdly, by inheriting the generalization capabilities of LLMs, LAURA is able to adapt to varying network topologies and SN distributions. Simulation results validate that LAURA significantly outperforms baseline methods, especially in large-scale scenarios with a growing number of SNs.

\section{System Model}
\subsection{Network Model}

\begin{figure}[t!]
\centering
        \includegraphics[width=0.485\textwidth]{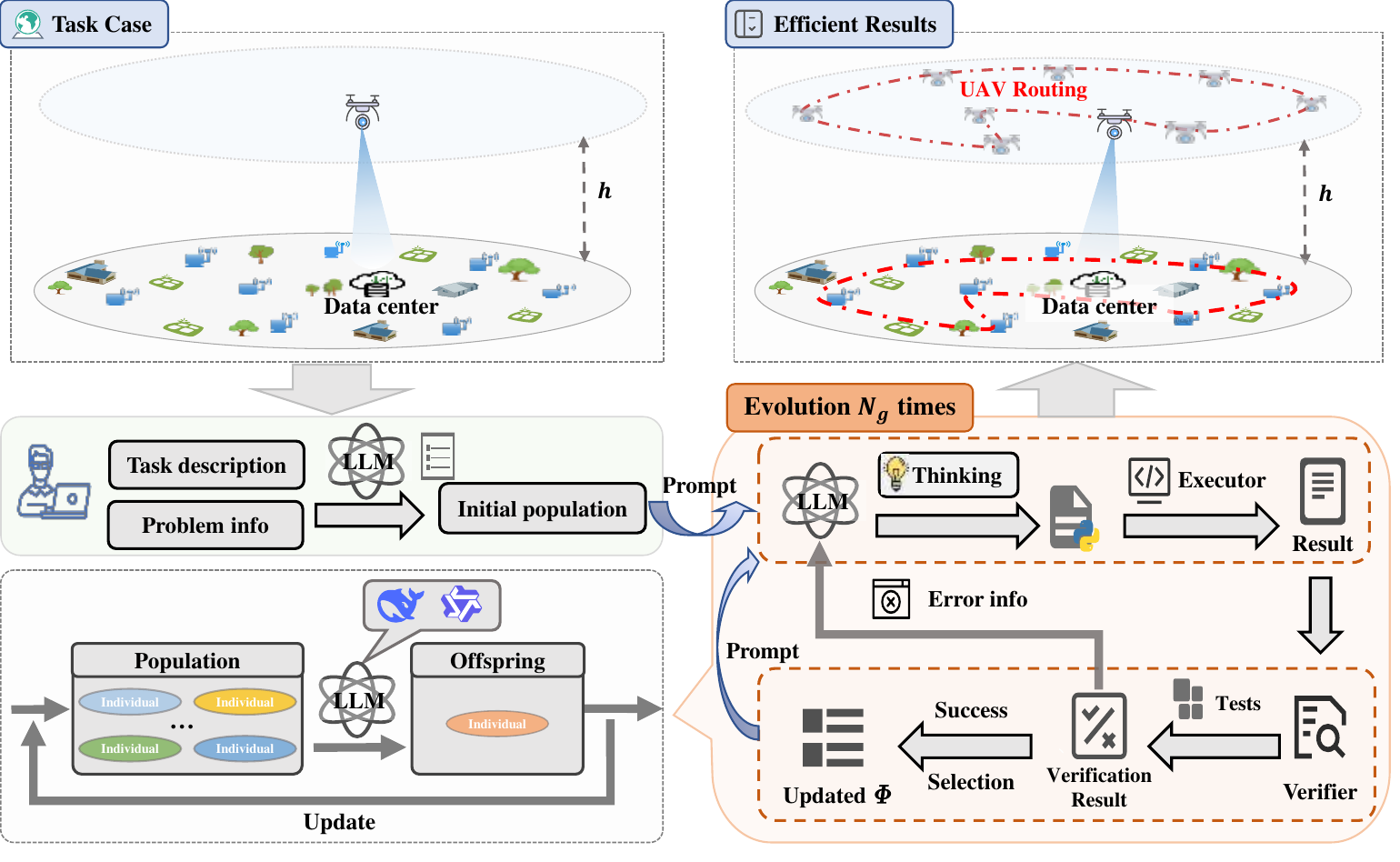}
        \caption{System model with UAV, a data center, and $N$ SNs, using LAURA to optimize UAV routing via LLM and evolutionary mechanisms.}
        \label{scenario}
\end{figure}

We consider a UAV-assisted WSN system, consisting of a data center $s_0$, a UAV, and $N$ ground SNs denoted by $\boldsymbol{\mathcal{S}}=\{s_1,s_2,\dots,s_N\}$. Each SN $s_i$, including the date center $s_0$, is stationary and positioned at coordinates $(x_i,y_i,0)$ in a three-dimensional Cartesian space. The UAV serves as a mobile data aggregator, tasked with collecting data from the SNs. For traceability, we assume that the UAV operates at a fixed altitude $h$ and maintains a constant speed $v$. The flight time $t_{i,j}$ required for the UAV to travel between any two SNs $s_i$ and $s_j$ is then given by
\begin{flalign}
t_{i,j} &= \tfrac{d_{i,j}}{v} =  \tfrac{1}{v} || s_{i} - s_{j} ||   \notag \\ 
  & =  \tfrac{1}{v} \sqrt{(x_{i} - x_{j})^2 + (y_{i} - y_{j})^2}, \quad \forall i \neq j.
\end{flalign}
The UAV initiates its mission at the data center $s_0$, sequentially visits each SN just once to collect data, and ultimately returns to $s_0$ for data offloading and further analysis. The complete UAV flight routing sequence $\boldsymbol{\mathcal{C}}$ is represented as
\begin{flalign}
\boldsymbol{\mathcal{C}}=[c_0,c_1,c_2,\dots,c_N,c_{N+1}],
\end{flalign}
where $c_i, \forall i \in \{1,2,\dots,N\}$ denotes the $i$-th visited SN in this sequence, with $c_0$ and $c_{N+1}$ representing the UAV's start and end points, respectively, both located at $s_0$.

\subsection{Data Transmission}
Upon reaching an SN $s_i,\forall i \in \{1, 2, \dots, N\}$, the UAV establishes a communication link and begins collecting data. Since the time required for link establishment is considered negligible compared to the data transmission and UAV flight time, it is omitted from the analysis, as in \cite{AoI_Scenario}. We also assume a line-of-sight (LoS) dominated wireless channel between the UAV and each SN. The channel power gain is modeled as $\frac{g_i}{h ^2}$, where $g_i$ denotes the channel gain at the reference distance. Under a fixed transmission power $P_0$, the achievable data rate from SN $s_i$ to the UAV is given by
\begin{flalign}
   R_i = W \log_{2} \left( 1 + \frac{P_0 g_i}{\sigma^2 h^2} \right),
\end{flalign} 
where $W$ is the system bandwidth and $\sigma ^2$ is the received noise power. Let $D_i$ be the size (in bits) of the data stored at $s_i$. Then, the duration required to fully transmit $D_i$ to the UAV, assuming a continuous and uninterrupted link, is
\begin{flalign}
  \tau_i = \frac{D_i}{W \log_{2} \left( 1 + \frac{P_0 g_i}{\sigma^2 h^2} \right)} = \frac{D_i}{R_i},
\end{flalign}
which characterizes the data collection latency from SN $s_i$ under a constant-rate and interference-free channel assumption.


\subsection{Age-of-Information (AoI) Metric}
The AoI measures the freshness of the data, which depends on both the data transmission time and the UAV’s subsequent flight time. Let $A_i(t)$ denote the AoI for the data transmitted from SN $s_i$ at time $t$, which is defined as
\begin{flalign}
  A_i(t)=(t-t_i)u(t-t_i), i=1, \dots, N,
\end{flalign}
where $t_i$ is the time at which $s_i$ generates and uploads its data. $u(x)=\max \{x,0 \}$ is the Heaviside step function, ensuring that AoI is non-negative.

Furthermore, let $t_{\textrm{end}}$ be the time required for the UAV to complete its entire routing. For the $i$-th SN in the flight routing sequence $\boldsymbol{\mathcal{C}}$, the corresponding AoI can be calculated by
\begin{flalign}
  A_{c_i}=A_{c_i}(t_{\textrm{end}})=\sum\nolimits_{j={i}}^{N}\tau_{c_j}+\sum\nolimits_{j=i}^{N}t_{c_j,c_{j+1}} ,
\end{flalign}
where $\sum_{j={i}}^{N}\tau_{c_j}$ is the cumulative data transmission time for the SNs from $c_i$ onwards in the routing sequence, and $\sum_{j=i}^{N}t_{c_j,c_{j+1}}$ represents the overall flight time required by the UAV to complete the data collection from $c_i$, visit all remaining SNs in the sequence, and finally return to the data center $s_0$. 
The UAV’s energy consumption is not explicitly modeled under the assumption of sufficient onboard energy or availability of recharging, and is thus omitted here.


\section{Problem Formulation}
To maintain data freshness, we aim to design a UAV flight routing sequence, i.e., $\boldsymbol{\mathcal{C}}$, that minimizes the maximum AoI across all SNs. By incorporating both data transmission durations and UAV travel time, the proposed formulation seeks to determine an efficient routing strategy for timely data collection. Thus, the problem is mathematically formulated as\begin{subequations} \label{problem_1}
\begin{flalign}
    \textbf{P}_{1}: \min \limits_{\boldsymbol{\mathcal{C}}} \,\, & \max_{i}  \,\, A_{c_i}   \notag \\
    \mathrm{s.t.} \,\, 
    & \tau_{c_k} > 0, \forall k, \label{p1_c1}\\
    & t_{c_k,c_{k+1}} > 0 , \forall k , \label{p1_c2}\\
    &  c_{0} = c_{N+1} = s_0 , \label{p1_c3} \\ 
    &c_i \neq c_j, \forall i \neq j , \label{p1_c4}  
\end{flalign}
\end{subequations}
where $t_{c_i,c_j}= \frac{d_{c_i,c_j}}{v}$ denotes the UAV travel time between SNs, which is based on their Euclidean distance. Constraints \eqref{p1_c1} and \eqref{p1_c2} respectively represent the UAV's data collection duration and travel time, which should be greater than zero. Constraint \eqref{p1_c3} ensures that the UAV flight starts and ends at the data center $s_0$, and constraint \eqref{p1_c4} guarantees that each SN is visited only once.

Given the objective of minimizing the maximum AoI across all SNs, it is essential to find an efficient flight routing $\boldsymbol{\mathcal{C}}$ that ensures timely data collection. This problem is inherently combinatorial, as the UAV's visit order directly affects AoI minimization. To provide further insight into this problem, we introduce Lemma \ref{lem_1}, which establishes a fundamental property of AoI in UAV routing.

\begin{Lem} \label{lem_1}
For any flight routing, the AoI of SNs decreases monotonically with the data collection order $\{c_1,c_2,\dots,c_N\}$.
\end{Lem}

\begin{proof}

Since the AoI of each SN increases with its waiting time, we aim to show that if $i < j$, then the SN visited earlier, $c_i$, has a higher AoI than $c_j$, i.e.,
\begin{flalign}
    A_{c_i} - A_{c_j} > 0.
\end{flalign}
To quantify it, consider the AoI difference between $c_i$ and $c_j$:
\begin{flalign}
A_{c_i} - A_{c_j} & = \left( \sum\nolimits_{k=i}^{N} \tau_{c_k} + \sum\nolimits_{k=i}^{N} t_{c_k,c_{k+1}} \right) \notag \\
& - \left( \sum\nolimits_{k=j}^{N} \tau_{c_k} + \sum\nolimits_{k=j}^{N} t_{c_k,c_{k+1}} \right) \notag \\
& = \sum\nolimits_{k=i}^{j-1} \left( \tau_{c_k} + t_{c_k,c_{k+1}} \right).
\end{flalign}
where $\tau_{c_k} \geq 0$ and $t_{c_k, c_{k+1}} \geq 0$ for all $k$. Hence, the summation is strictly positive for $i < j$, and we conclude that $A_{c_i} > A_{c_j}, i < j$. The proof is complete.


\end{proof}


According to Lemma \ref{lem_1}, the AoI decreases monotonically as the UAV visits the SNs sequentially. Hence, the objective of problem $\textbf{P}_1$ can be replaced by $A_{c_1}$ equivalently, where
$A_{c_1}=\ \sum\nolimits_{k=1}^{N} \tau_{c_k} + \sum\nolimits_{k=1}^{N} t_{c_k,c_{k+1}}.$ Since $\sum\nolimits_{k=1}^{N} \tau_{c_k}$ is constant, the problem $\textbf{P}_1$ is equivalent to a minimum-route problem, i.e., 
\begin{subequations}
\begin{flalign}
    \textbf{P}_2: \min \limits_{\boldsymbol{\mathcal{C}}}  \, & \sum\nolimits_{k=1}^{N} t_{c_k, c_{k+1}}\notag \\
    \mathrm{s.t.} \,\,\, & \eqref{p1_c2}, \eqref{p1_c3}, \eqref{p1_c4}.
\end{flalign}
\end{subequations}

At this point, the problem $\textbf{P}_2$ can be viewed as a TSP problem, where the goal is to find the shortest possible routing. However, solving problem $\textbf{P}_2$ is so challenging because there are $N!$ possible routings to consider, making it extremely demanding on computational resources, particularly for large values of $N$ \cite{AoI_dp}. Existing methods that rely on expert-designed rules need much specialized knowledge to work well. To improve on this, we propose the LAURA method, utilizing LLM that has been pre-trained with extensive knowledge, which can solve the problem more efficiently.

\section{Proposed Algorithm}

LLMs, which are transformer-based architectures pre-trained on vast dataset, exhibit robust reasoning capabilities through self-attention and contextual pattern recognition \cite{LLM_Ability}. Their proven efficacy in solving complex algorithmic tasks (e.g., mathematics, programming) motivates our integration of an LLM with evolutionary mechanisms to address the UAV routing problem $\textbf{P}_2$ \cite{LLM_ref}. Our proposed LAURA framework synergizes LLM's predictive reasoning with genetic evolutionary principles. As depicted in Fig. \ref{scenario}, the algorithm iteratively refines solutions through three phases:

\subsubsection{\textbf{Initialization}} 

Initialization involves creating an initial population, which forms the basis for the subsequent evolutionary process. To avoid the need for expert knowledge, we utilize an LLM to generate this initial set. As illustrated in Fig. \ref{Prompt}, we provide the LLM with a prompt $\Psi(\boldsymbol{\mathcal{S}})$ that contains the task description and specific initialization hints. The code extracted from the LLM's response is then executed to produce the initial population $\boldsymbol{\varPhi}$.

Let $\text{LLM}(\cdot)$ denote the result generated by the LLM. Define $\boldsymbol{\phi}^k = [\boldsymbol{\varphi}^k, \varOmega(\boldsymbol{\varphi}^k)]$ as an individual consisting of two parts: $\boldsymbol{\varphi}^k = [\varphi_0, \varphi_1, \dots ,\varphi_N,\varphi_{N+1}]$, representing a UAV routing, and $\varOmega(\boldsymbol{\varphi}^k)$, indicating the maximum AoI for that routing. Therefore, the initial population $\boldsymbol{\varPhi}$ is given by
\begin{flalign}
& \boldsymbol{\varPhi}  =\text{LLM}(\Psi(\boldsymbol{\mathcal{S}})) = [\boldsymbol{\phi}^1, \cdots, \boldsymbol{\phi}^K]. \label{p_phi}
\end{flalign}

\subsubsection{\textbf{Evolution}}

After initialization, the population evolves through an iterative process aimed at refining the quality of each individual. Over $N_{\text{g}}$ iterations, the LLM generate $N_{\text{p}}$ individuals. This evolution comprises three stages: selection, generation, and verification.

\textit{Selection Stage:} To build a diverse parent individual set $\boldsymbol{\varTheta}_\ell$ from the current population $\boldsymbol{\varPhi}_\ell$ at iteration $\ell<N_{\text{g}}$ and
avoid premature convergence to suboptimal solutions caused by reduced population diversity,
we randomly select $N_{\text{p}}$ individuals to form $\boldsymbol{\varTheta}_\ell$. The selection is expressed as
\begin{flalign}
    \boldsymbol{\varTheta}_\ell = \Xi_K^{N_{\text{p}}} \circledast \boldsymbol{\varPhi}_\ell,
\end{flalign}
where $\Xi_K^{N_{\text{p}}}$ represents the selection of $N_{\text{p}}$ individuals from the $K$ ones, and $\circledast$ denotes the operation of doing $\Xi_K^{N_{\text{p}}}$ in $\boldsymbol{\varPhi}_\ell$.

 \textit{Generation Stage:} Using the parent set $\boldsymbol{\varTheta}_\ell$, we design an evolution prompt $\Psi(\boldsymbol{\mathcal{S}},\boldsymbol{\varTheta}_\ell)$ that includes task descriptions, parent individuals, and specific hints. This guides the LLM to generate offspring individuals via crossover operations, exploring the solution space efficiently. The offspring individual $\boldsymbol{\theta}_\ell$ is thus obtained by 
    \begin{flalign}
    \boldsymbol{\theta}_\ell  = \text{LLM}(\Psi(\boldsymbol{\mathcal{S}},\boldsymbol{\varTheta}_\ell)) = \left[ \boldsymbol{\vartheta}_\ell, \varOmega(\boldsymbol{\vartheta}_\ell) \right], \label{evo_theta}
    \end{flalign}
where $\boldsymbol{\vartheta}_\ell = \left[\vartheta_0^\ell, \vartheta_1^\ell, \dots, \vartheta_N^\ell,\vartheta_{N+1}^\ell \right]$.

\textit{Verification Stage:} This stage validates the correctness of LLM-generated individuals and reduces errors. By executing the generated hybrid evolutionary algorithm code, we obtain $\boldsymbol{\vartheta}_\ell$ and $\varOmega(\boldsymbol{\vartheta}_\ell)$, and then verify them against the following equations:
\begin{subequations}
\label{verifier}
\begin{empheq}[left=\empheqlbrace]{align}
& \vartheta_0^\ell = \vartheta_{N+1}^\ell = s_0, \label{verifier_1} \\
& \left\{\vartheta_1^\ell, \dots, \vartheta_N^\ell \right\} \cap \boldsymbol{\mathcal{S}} = \boldsymbol{\mathcal{S}}, \label{verifier_2} \\
& \varOmega(\boldsymbol{\vartheta}_\ell) = \sum\nolimits_{i=1}^N \tau_{\vartheta_i^\ell} + \sum\nolimits_{i=1}^{N} \tfrac{d_{\vartheta_i^\ell,\vartheta_{i+1}^\ell}}{v}, \label{verifier_3}
\end{empheq}
\end{subequations}
where \eqref{verifier_1} guarantees that the routing $\boldsymbol{\vartheta}_\ell$ starts and ends at the data center $s_0$, \eqref{verifier_2} confirms that the UAV visits each SN in $\boldsymbol{\mathcal{S}}$ just once, and \eqref{verifier_3} ensures the accuracy of calculating the maximum AoI. If all conditions in \eqref{verifier} are satisfied, the verification is successful, and $\boldsymbol{\theta}_\ell$ is added to the current population $\boldsymbol{\varPhi}_\ell$, i.e.,
\begin{flalign}
    \boldsymbol{\varPhi}_\ell^{\text{new}} = \boldsymbol{\varPhi}_\ell \cup \boldsymbol{\theta}_\ell.
\end{flalign}
Otherwise, an error is returned and the process repeats until $M$ attempts are completed.

\subsubsection{\textbf{Update}} 
The update process aims to form the next-generation population from the current one. To maintain population quality, the individual with the lowest fitness is removed. The fitness function is defined as
\begin{flalign}
   f(\boldsymbol{\theta})  = e^{- \varOmega(\boldsymbol{\vartheta})},
\end{flalign}
where a lower maximum AoI yields a higher fitness value, thus favoring solutions with lower AoI.

The population update then proceeds as follows:
\begin{flalign}
    \boldsymbol{\varPhi}_{\ell+1} & = \Xi_{K+1}^K \circledast \boldsymbol{\varPhi}_\ell^{\text{new}}, \label{l_1}  
\end{flalign}
where $\boldsymbol{\varPhi}_{l+1}$ is the updated population at iteration $(\ell+1)$, obtained by selecting the top $K$ individuals from the expanded population $\boldsymbol{\varPhi}_\ell^{\text{new}}$. 

Consequently, we obtain the optimized population $\boldsymbol{\varPhi}^\star$, which contains the top $K$ individuals. From this optimized population, the efficient solution $\{\boldsymbol{\mathcal{C}}^\star,A_{c_1}^\star\}$ is given by
\begin{flalign}
  \{ \boldsymbol{\mathcal{C}}^\star,A_{c_1}^\star\} = \left\{\mathop{\arg}\min\limits_{\boldsymbol{\phi} \in \boldsymbol{\varPhi}^\star} \,\varOmega(\boldsymbol{\varphi}) ,\min\limits_{\boldsymbol{\phi} \in \boldsymbol{\varPhi}^\star} \varOmega(\boldsymbol{\varphi}) \right\}.
\end{flalign}

\begin{algorithm}[t!]
  \caption{The framework of our proposed LAURA}
  \label{LAURA_algorithm}
  \begin{algorithmic}[1] 
    \REQUIRE The problem $\textbf{P}_2$, the population size $K$, the parent size $N_{\text{p}}$, the maximum attempts $M$, the new individuals size $N_{\text{g}}$ generated by the LLM, the prompt $\Psi$.    
    \STATE  Initialize the population $\boldsymbol{\varPhi}= \text{LLM}(\Psi(\boldsymbol{\mathcal{S}}))$ by \eqref{p_phi}.    
    \FOR{$\ell =1$ \TO $N_{\text{g}}$}
      \STATE  Select the parent set $\boldsymbol{\varTheta}_\ell$ from $\boldsymbol{\varPhi}_\ell$.
      \STATE  Design the evolution prompt $\Psi(\boldsymbol{\mathcal{S}},\boldsymbol{\varTheta}_\ell)$ based on the parent set. \label{Step: LLM}
      \STATE  Provide the designed prompt to the LLM and obtain the offspring individual $\boldsymbol{\theta}_\ell$ by \eqref{evo_theta}.
      \STATE  Verify $\boldsymbol{\theta}_\ell$ by \eqref{verifier}.
      \IF{verification successful}
        \STATE  Add $\boldsymbol{\theta}_\ell $ to the current population $\boldsymbol{\varPhi}_\ell$.
      \ELSE
        \STATE  Return the error and repeat from step \ref{Step: LLM} until $M$ attempts are completed.
      \ENDIF
      \STATE  Update next population $\boldsymbol{\varPhi}_{\ell+1}$ by \eqref{l_1}.
    \ENDFOR

    \STATE Obtain the efficient result: $\{\boldsymbol{\mathcal{C}}^\star, A_{c_1}^\star \}$.
    \ENSURE The solution: $\{\boldsymbol{\mathcal{C}}^\star, A_{c_1}^\star \}$.
  \end{algorithmic}
\end{algorithm}
\vspace{-0.15cm}

By utilizing an LLM directly as crossover operators, LAURA efficiently explores the solution space, as detailed in Algorithm \ref{LAURA_algorithm}. Moreover, considering the diverse and unknown structures of LLMs, we denote the complexity of an LLM as $N_{\textrm{LLM}}$. The computational complexity of LAURA is approximately given by $\mathcal{O}(M\times N_{\text{g}}\times N_{\textrm{LLM}})$. While LAURA may not always be the most time-efficient approach, it significantly reduces computational resource demands compared to dynamic programming and eliminates the need for expert knowledge required by traditional heuristic methods.


\section{Numerical Results}
We provide numerical results to validate LAURA's performance in the UAV-assisted WSN. In the simulation, the SNs are randomly and uniformly distributed within a circle area of 3000 m radius and the data center $s_0$ is located at the origin $\left(0,0\right)$. The UAV operates at a height of 30 m with a speed $v$ of 10 $\text{m/s}$. Unless otherwise stated, we set the key parameters as: UAV's transmit power $P_0 = 0.3 \,\text{W}$, data size $D_i = 0.5 \,\text{Mb}$, system bandwidth $W = 1 \,\text{MHz}$, channel gain at 1 meter $g_i=-50\text{dB}$, and noise power $\sigma ^2 = -110\, \text{dBm}$ \cite{AoI_dp}. For LAURA, we set a population size of $K = 10$, select $N_{\textrm{p}} = 5$ parent individuals, repeat maximum $M=3$ attempts and generate $N_{\text{g}} = 10$ new individuals using an LLM, i.e., Deeepseek v3-671B, QwQ-32B and Deepseek R1-32B (a distilled version of Deepseek R1).

For comparison, we simulate LEDMA \cite{ISAC_LLM}, Genetic, Greedy, and Random algorithms as benchmarks. In particular, the LEDMA algorithm generates UAV routing solutions using an LLM directly and does not provide feedback on verification failures for secondary attempts. We evaluate performance across three scenarios with $N=\{20, 30, 40\}$ SNs, respectively. In each scenario, ten random cases are generated and solved in five independent runs using different algorithms.

\begin{figure}[t!]
\centering
        \includegraphics[height=0.25\textheight, keepaspectratio=true]{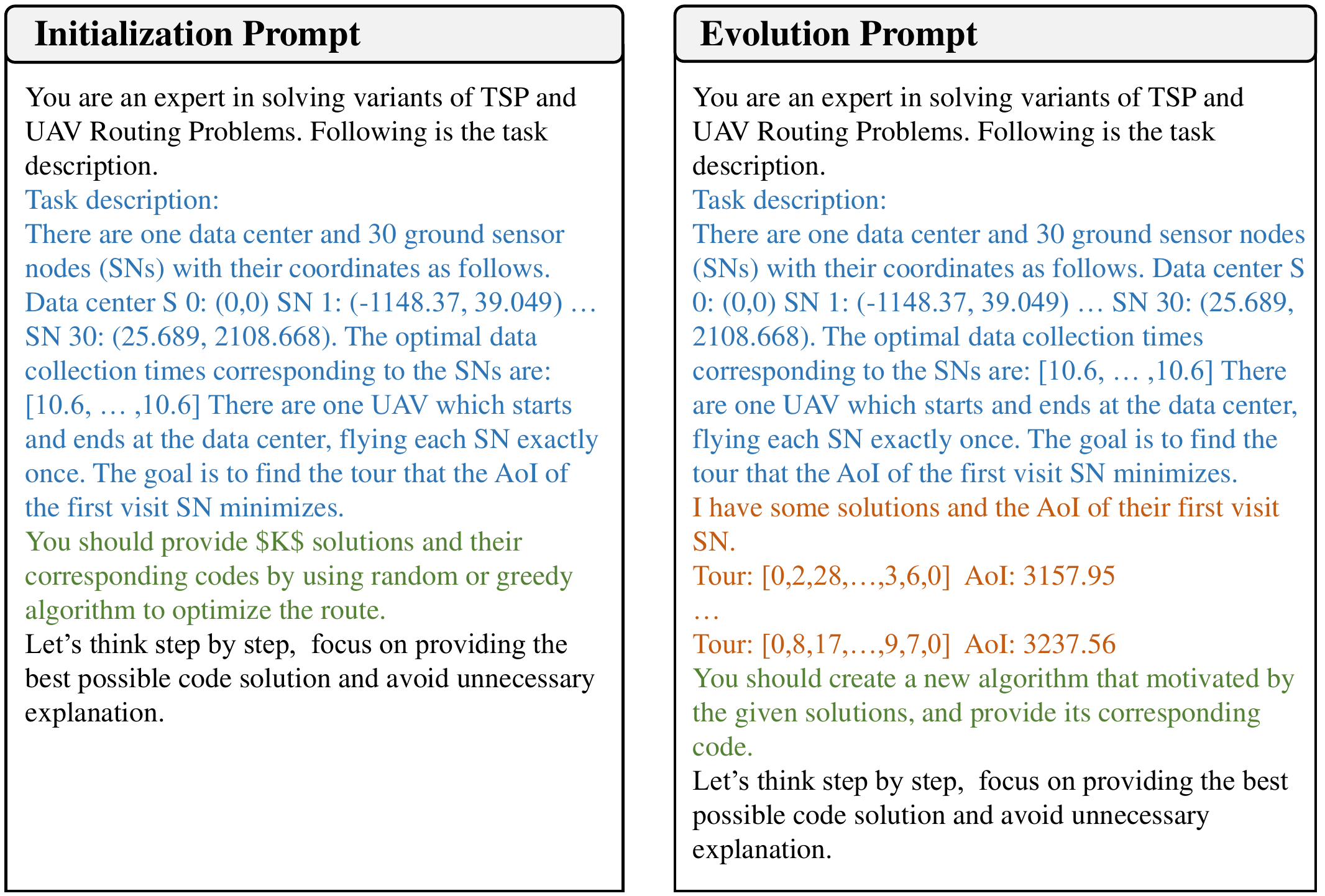}
        \caption{Instance Prompts used in the initialization and evolution of LAURA: \textcolor[RGB]{55,123,185}{a Task description}, \textcolor[RGB]{208,123,65}{Parent solutions}, and \textcolor[RGB]{137,168,115}{Prompt-specific hints}.}
        \label{Prompt}
        \vspace{-0.5cm}
\end{figure}

\begin{figure*}[t!]
  \centering
  \subfigure[The average of AoI.]{%
      \includegraphics[height=0.18\textheight, keepaspectratio=true]{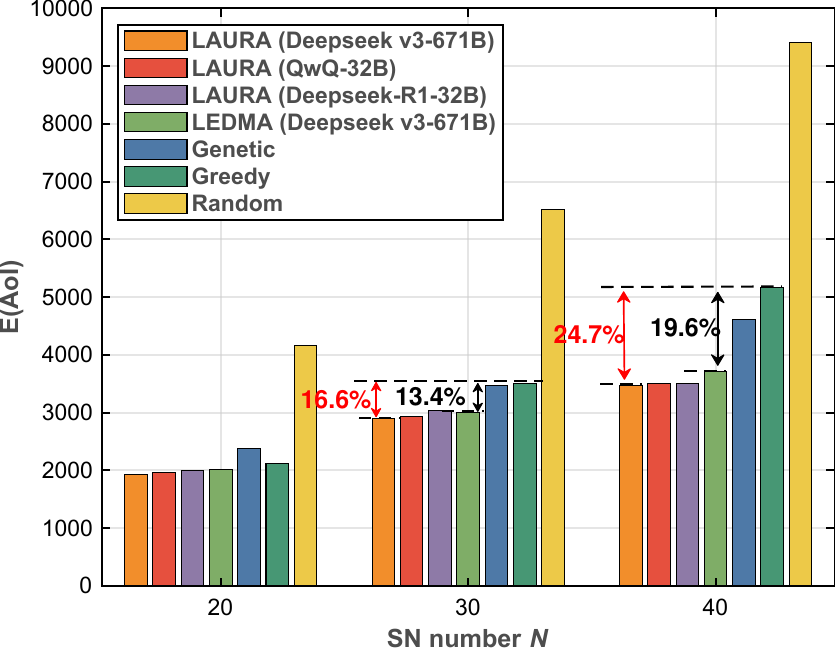}%
      \label{Average_AoI}%
  }
  \hspace{0.5cm} 
  \subfigure[The variance of AoI.]{%
      \includegraphics[height=0.186\textheight, keepaspectratio=true]{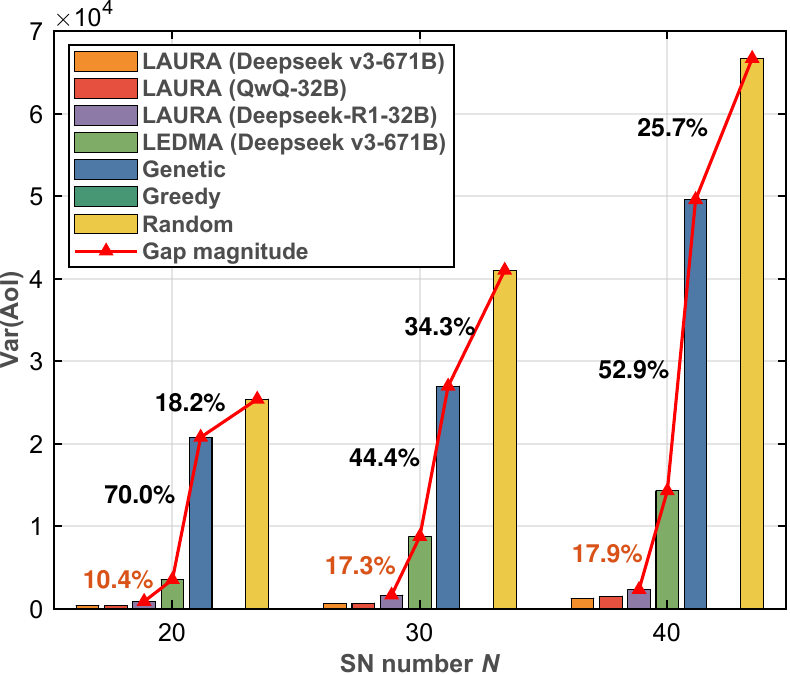}%
      \label{Variance_AoI}
  }
  \hspace{0.5cm} 
  \subfigure[The hallucination rate $\boldsymbol{\epsilon}$.]{%
      \includegraphics[height=0.18\textheight, keepaspectratio=true]{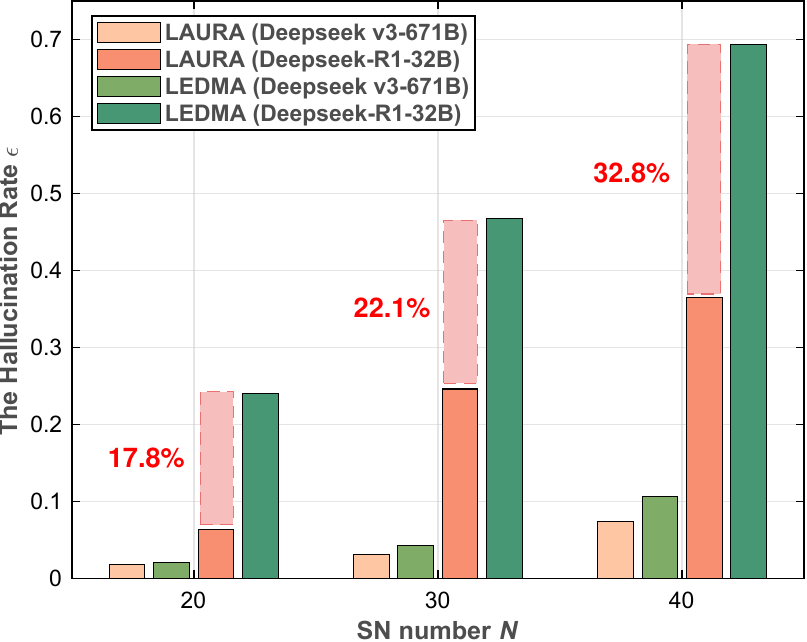}%
      \label{Error_rate}%
  }
  \caption{The mean and variance comparison of $A_{c_1}$ for LAURA, LEDMA \cite{ISAC_LLM}, Genetic, Greedy and Random algorithms, and the evaluation of the hallucination rate $\boldsymbol{\epsilon}$ for LEDMA and LAURA, across $N=20$, $30$, and $40$ SNs, respectively.}
  \label{Result_1}
\end{figure*}

Fig. \ref{Average_AoI} shows the mean of $A_{c_1}$ for LAURA, LEDMA, Genetic, Greedy and Random algorithms across $N=20$, $30$, and $40$ SNs, respectively. It is seen that LAURA achieves a lower maximum AoI than other algorithms. Specifically, LAURA (Deepseek-v3) reduces the mean AoI by 16.6\% and 24.7\% compared to the Greedy algorithm at $N=30$ and $N=40$, respectively. The performance gaps between them become large with the increment of $N$, due to the integration 
of LLM-guided exploration and evolutionary mechanisms for preserving high-quality solutions.
Meanwhile, LEDMA (Deepseek-v3) achieves reductions of 13.4\% and 19.6\%, demonstrating LAURA’s superiority in WSNs. 
Moreover, the Genetic algorithm underperforms due to reliance on manual expert knowledge and even performs worse than the Greedy one when $N=20$, revealing the limitations of traditional heuristic methods.

To evaluate algorithm stability, we also plot the variance of $A_{c_1}$ across $N=\{20, 30, 40\}$ SNs for different algorithms in Fig. \ref{Variance_AoI}. It is observed that our LAURA can maintain a lower variance of AoI compared to other algorithms, even at larger $N$, indicating its stability. However, LAURA with DeepSeek-R1-32B shows slightly poorer performance than LAURA with other LLMs, partly because its smaller parameter count impacts model effectiveness. Moreover, the performance gap between different algorithms becomes larger as $N$ increases. Particularly, the variance gap between LAURA (Deepseek-R1-32B) and LEDMA (Deepseek v3) reaches 10.4\%, 17.3\% and 17.9\% at $N=20, 30$ and 40, respectively.

Fig. \ref{Error_rate} compares the model hallucination rate $\boldsymbol{\epsilon}$ of LAURA and LEDMA. Specifically, LLMs in LEDMA hallucinate when the outputs fail to meet the criteria in \eqref{verifier_1} and \eqref{verifier_2}, whereas in LAURA, hallucinations occur when the outputs do not satisfy all conditions in \eqref{verifier}. We can see that $\boldsymbol{\epsilon}$ of LAURA is lower than that of LEDMA. Moreover, the gaps in $\boldsymbol{\epsilon}$ between LAURA (Deepseek-R1-32B) and LEDMA (Deepseek-R1-32B) are 17.8\%, 22.1\% and 32.8\% at $N=20, 30$ and 40, respectively.

\begin{figure*}[t!]
  \centering
  \subfigure[UAV routing by LAURA.]{%
      \includegraphics[height=0.14\textheight]{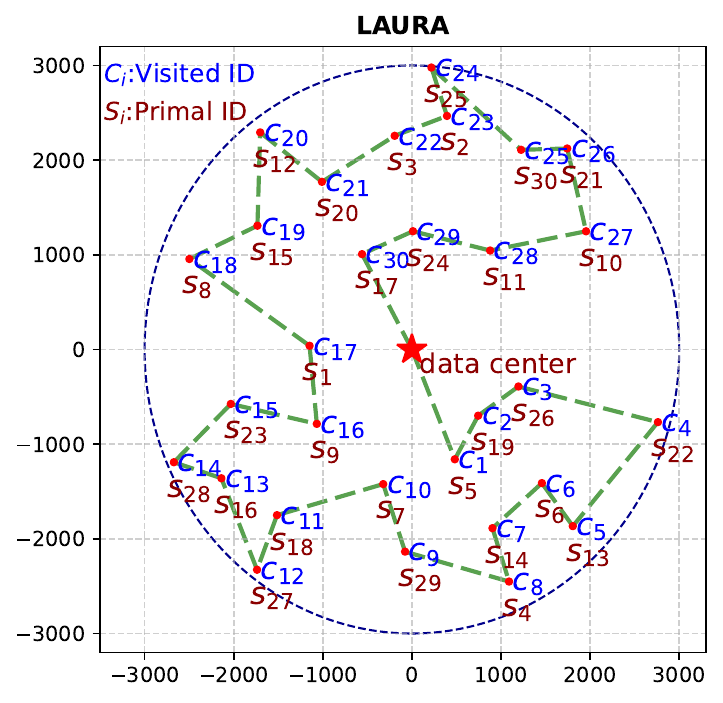}%
      \label{N_30_1_(LAURA)} 
  }
  \subfigure[UAV routing by LEDMA.]{%
      \includegraphics[height=0.14\textheight]{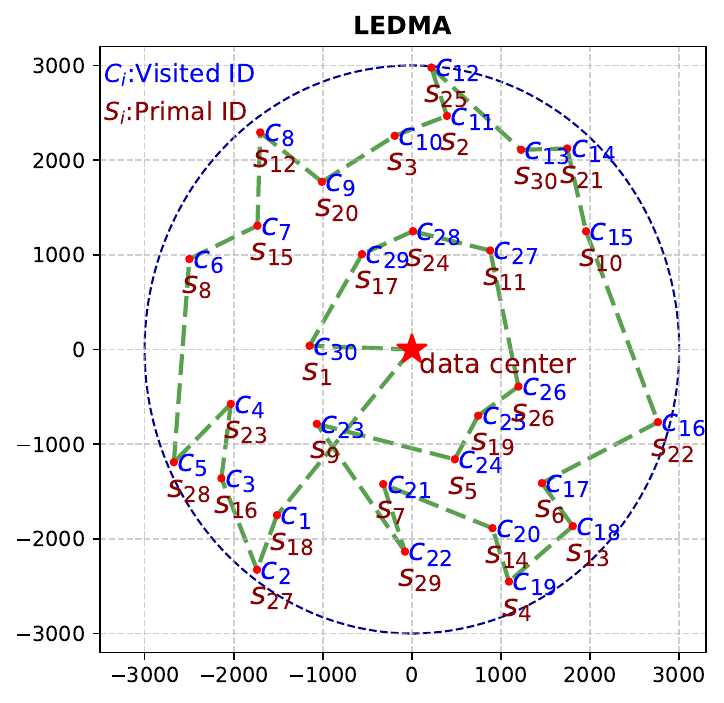}%
      \label{N_30_1_(LEDMA)}%
  }
  \subfigure[UAV routing by Greedy.]{%
      \includegraphics[height=0.139\textheight]{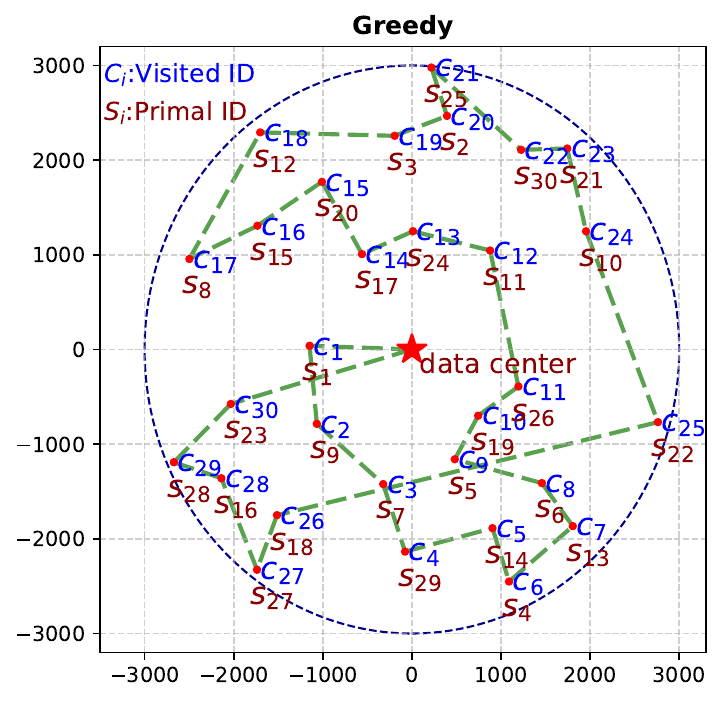}%
      \label{N_30_1_(Greedy)}%
  }
  \subfigure[UAV routing by Genetic.]{%
      \includegraphics[height=0.139\textheight]{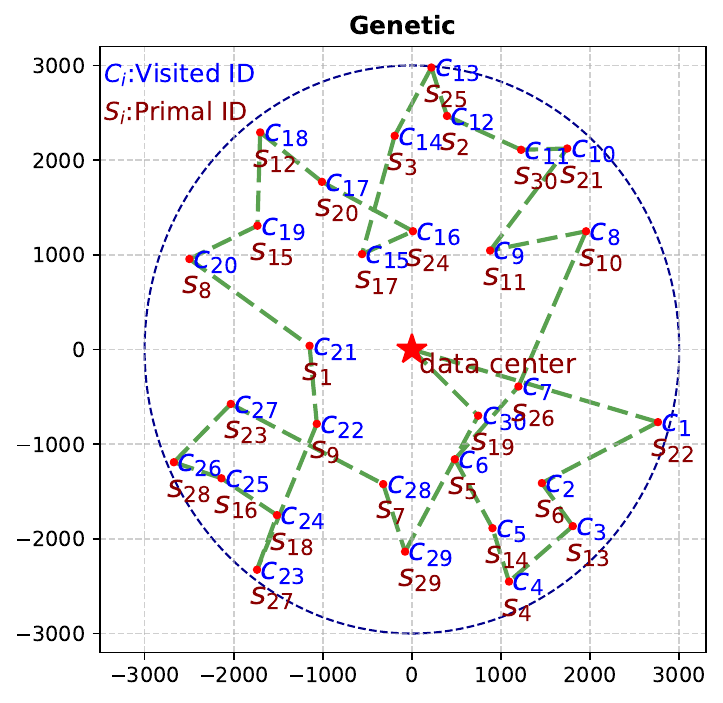}%
      \label{N_30_1_(Genetic)}%
  }
  \subfigure[UAV routing by Random.]{%
      \includegraphics[height=0.139\textheight]{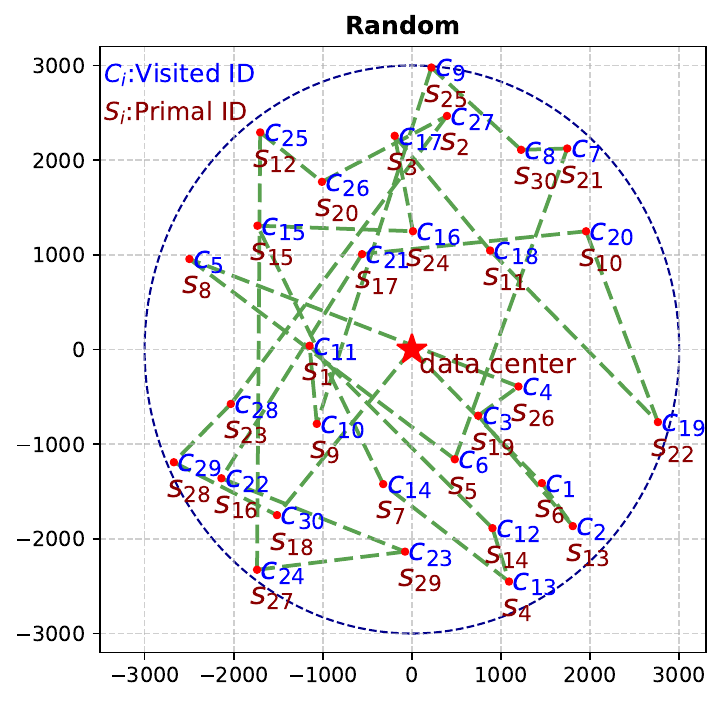}%
      \label{N_30_1_(Random)}%
  }
  \caption{The comparison of UAV routing results obtained by LAURA, LEDMA \cite{ISAC_LLM}, Genetic, Greedy, and Random algorithms with $N=30$ SNs.}
  \label{routing}
  \vspace{-0.3cm}
\end{figure*}

Fig. \ref{routing} compares the UAV routing performance of the proposed LAURA algorithm against several baselines for $N=30$. We can see that LAURA optimizes the UAV flight routing more effectively, avoiding overlapping areas in Fig. \ref{N_30_1_(LAURA)}. While the Genetic algorithm converges to a local optimum, yielding a suboptimal sequence. LEDMA, implemented with Deepseek v3, achieves slight improvements over the Genetic algorithm but still underperforms compared to LAURA. The Greedy approach, based on nearest selection, leads to inefficient ordering, while the Random strategy performs the worst due to its lack of coordination.

\section{Conclusion}
This study has investigated a UAV-assisted WSN system aimed at minimizing the maximum AoI across ground SNs. We have formulated the optimization problem and proposed LAURA, a novel method that combines evolutionary mechanisms with LLM-guided crossover operations to improve solution space exploration. Numerical results have confirmed that LAURA outperforms benchmark algorithms.



\bibliographystyle{IEEEtran}
\bibliography{IEEEabrv.bib,REF}

\end{document}